\documentclass[11pt,a4paper]{article}

\usepackage[T1]{fontenc}
\usepackage[utf8]{inputenc}
\usepackage{lmodern}
\usepackage[margin=2.5cm]{geometry}
\usepackage{amsmath,amssymb,amsthm}
\usepackage{graphicx}
\usepackage{booktabs}
\usepackage{multirow}
\usepackage{array}
\usepackage{algorithm}
\usepackage{algpseudocode}
\usepackage{xcolor}
\usepackage{url}
\usepackage{xurl}
\usepackage[hidelinks]{hyperref}
\usepackage{authblk}
\usepackage[section]{placeins}
\usepackage{setspace}
\usepackage{etoolbox}

\AtBeginEnvironment{algorithmic}{\setstretch{1.35}}

\newtheorem{lemma}{Lemma}

\newcommand{\fl}{\mathrm{fl}}
\newcommand{\rint}{\operatorname{rint}}
\newcommand{\sym}{\operatorname{sym}}

\title{Emulating DGEMM and ZGEMM with INT8 Matrix Engines on Intel Arc GPUs:\\
Ozaki Scheme II on the Xe2 Architecture}

\author[1]{Rostislavs Rostovskis\thanks{Corresponding author: \texttt{rostislavs.rostovskis@cfi.lu.lv}}}
\author[1]{Aleksandrs Platonenko}
\affil[1]{Institute of Solid State Physics, University of Latvia, Riga, Latvia}
\date{\today}

\begin{document}
\maketitle

\begin{abstract}
Double-precision matrix multiplication, real (DGEMM) and complex (ZGEMM), dominates
the run time of dense linear algebra and electronic-structure codes. On affordable
GPUs, however, FP64 throughput is a small and shrinking fraction of the
throughput of the low-precision matrix engines. Ozaki scheme~II closes this gap by reducing a floating-point matrix product
to a handful of exact INT8 products followed by a Chinese-remainder reconstruction;
the number of products sets the accuracy. We present what is, to our knowledge, the first
implementation of Ozaki scheme~II for Intel GPUs. Written in SYCL on top of oneMKL, it runs entirely
on the device and exposes the full BLAS \texttt{dgemm}/\texttt{zgemm} interface,
including transposition, conjugation, $\alpha$, $\beta$ and arbitrary problem
sizes through memory-aware tiling. Making the scheme fast on the Xe2 (Battlemage)
architecture required keeping FP64 arithmetic out of all memory-bound kernels, a
branch-free residue computation, an FP32 evaluation of the CRT quotient with a
proven safety margin, 16-bit limb arithmetic in the reconstruction, and
work-arounds for performance cliffs of the vendor INT8 GEMM.
On an Intel Arc Pro B60 (24\,GB) the emulated DGEMM reaches 17.3, 12.1 and
9.3\,TFLOP/s with $s=8$, 12 and 16 moduli at $N=8192$, and up to 21.1\,TFLOP/s
overall, i.e.\ 12--28$\times$ the native FP64 DGEMM of the same board; 75--90\% of
the run time is spent in INT8 GEMM at $\approx$95\% of its attainable rate. The
emulated ZGEMM, which uses the 3M form exactly in integer arithmetic, reaches up to
27.4\,TFLOP/s and is 13.7--17.8$\times$ faster than native ZGEMM for $N\ge4096$ at
$s=16$. With $s\ge16$ the results are correctly rounded for well-scaled inputs and,
using an extended-width residue path ($s=20$--$24$), also for inputs with wide
exponent ranges, whereas native DGEMM on GPU or CPU rounds only 2--13\% of the
entries correctly.
We further report worst-case accuracy, phase breakdowns, energy, an ablation of the
pipeline-level decisions, a comparison with Ozaki scheme~I on the same hardware, and a
validated performance model for projecting the method to Intel's announced
Xe3P-based products.
\end{abstract}

\section{Introduction}
\label{sec:intro}

\paragraph{High-precision matrix multiplication.}
The general matrix product $C \leftarrow \alpha\,\mathrm{op}(A)\,\mathrm{op}(B)+\beta C$
(GEMM, a level-3 BLAS routine \cite{dongarra1990blas3}) is the kernel onto which
dense linear algebra is mapped: blocked LU, Cholesky and QR factorizations,
eigenvalue and singular-value solvers, and the High-Performance LINPACK benchmark
all perform the overwhelming majority of their floating-point operations in GEMM
\cite{dongarra2003linpack}. Computational chemistry and materials science are a
prominent example: self-consistent-field methods such as Hartree--Fock and
Kohn--Sham density functional theory repeatedly transform and multiply dense Fock,
overlap and density matrices, both in localized-basis and in plane-wave
implementations (e.g.\ \cite{kresse1996vasp}). With Bloch
$k$-points these matrices are complex Hermitian, and many-body perturbation methods
such as GW \cite{deslippe2012berkeleygw} are dominated by large complex products
(ZGEMM). Complex GEMM is equally central to frequency-domain electromagnetics,
quantum transport and quantum-circuit simulation.

These applications need double precision. Total-energy differences of a few meV on
totals of $10^3$--$10^5$\,Hartree, orbital orthogonalization, ill-conditioned
overlap matrices and long iterative processes all amplify rounding errors beyond
what single precision can absorb \cite{higham2002accuracy}. Mixed-precision
algorithms such as iterative refinement move part of the work to lower precision
\cite{haidar2018refinement,abdelfattah2021mixed}, but they still require an FP64
residual and do not replace GEMM itself.

\paragraph{The FP64 gap and its emulation.}
GPU vendors now spend most of their transistor budget on low-precision matrix
engines for machine learning, and FP64 is falling behind. On NVIDIA's current
Blackwell generation, both the GeForce RTX~50 series and the RTX PRO workstation
boards built on the same GB20x chips execute FP64 at $1/64$ of the FP32 rate
\cite{nvidia_rtx_blackwell}. Even in the data centre, the Blackwell Ultra
platform HGX~B300 delivers 10\,TFLOP/s of FP64, compared with 296\,TFLOP/s for
its predecessor HGX~B200 \cite{nvidia_hgx}. Emulating FP64 GEMM with exact
low-precision products is therefore attractive, and NVIDIA itself now ships such
an emulation in cuBLAS \cite{cuda1302_relnotes,brower2025cublasemu}.
Ozaki scheme~I \cite{ozaki2012} splits the significands into $L$ slices and needs
$O(L^2)$ slice products; it has been mapped to Tensor Cores
\cite{mukunoki2020tensorcores} and to INT8 engines
\cite{ootomo2024dgemm,uchino2025ozaki1int8}, and a related splitting recovers FP32
accuracy on Tensor Cores \cite{ootomo2022sgemm}. Ozaki scheme~II
\cite{ozaki2025scheme2} takes a different route. Instead of splitting the
significands, it represents the integerized matrices by their residues modulo $s$
small, pairwise coprime integers and multiplies the residue matrices exactly. Each
additional modulus adds about 8 bits of precision at the cost of one more INT8
product, so the work grows linearly rather than quadratically with the target
precision. The scheme has since been extended to FP8 \cite{uchino2026fp8}
and to complex products \cite{uchino2025complex}, and shown to be competitive on
CPUs with AMX and SVE2 \cite{kouya2026ozaki2cpu}.

\paragraph{Why Intel Arc.}
Intel's Xe2-based Arc Pro B60 offers 24\,GB of GDDR6 \cite{intel_arc_b60} at a
retail price of about US\$600 \cite{b60_price}, an unusually large memory for its
price class. Unlike the preceding Arc Alchemist generation, Xe2 has native FP64
support \cite{intel_gpu_optguide}, but only at $1/16$ of the FP32 rate: we measure
0.76\,TFLOP/s for DGEMM and 11.9\,TFLOP/s for SGEMM, while the XMX matrix engines
deliver $\approx$190\,TOP/s of INT8. This ratio of $\approx$250 between INT8 and
FP64 throughput lets the emulation pay off by an order of magnitude, and the large
memory accommodates the $O(sN^2)$ workspace of the scheme for large problems.
Intel has also announced, but not yet shipped, its Xe3P-based data-centre GPU
``Crescent Island'' \cite{intel_crescent_island}, now the Intel GPU D-Series. Its
first model, the D270, carries up to 480\,GB of LPDDR5X and is advertised for
workloads from low-precision inference to FP64 scientific computing
\cite{intel_d_series}. Both products thus combine large memory with strong
low-precision throughput, which makes them natural targets for Ozaki scheme~II.

\paragraph{Contributions.}
\begin{itemize}
\item A complete, device-resident DGEMM/ZGEMM emulation for Intel GPUs in SYCL
  and oneMKL, with the full BLAS interface (\texttt{N/T/C}, $\alpha$, $\beta$,
  leading dimensions), $2\le s\le 24$ moduli, memory-aware tiling and $k$-splitting
  (Sections~\ref{sec:method} and~\ref{sec:zgemm}).
\item The implementation choices that matter on Xe2: integer-only residue
  generation from the bits of IEEE doubles, a branch-free symmetric residue, an
  FP32 CRT quotient with a proven safety margin of $1/4$, 16-bit limb
  accumulation, correctly rounded conversion, and work-arounds for performance
  cliffs of the oneMKL INT8 GEMM and for VRAM over-subscription
  (Section~\ref{sec:arc}).
\item A ZGEMM in the exact 3M (Karatsuba) form of \cite{uchino2025complex},
  integrated into the same device-resident pipeline: the split kernel emits the
  residues of $A_r+A_i$ directly, and the CRT kernel combines the three products
  before reconstruction, so that a complex product costs $3s$ INT8 GEMMs and
  inherits the accuracy of the real case (Section~\ref{sec:zgemm}).
\item An extended-width residue path (scaling width $t$ up to 100 bits) that
  delivers correctly rounded DGEMM and ZGEMM for inputs with wide exponent ranges.
\item A systematic evaluation covering verification (unit tests of the kernels
  and element-wise checks of the error bound), accuracy against a double-double
  reference for typical and worst-case inputs, performance with a breakdown into
  phases, energy efficiency, an ablation study, a comparison with Ozaki scheme~I on
  the same INT8 engine, and a validated performance model
  (Sections~\ref{sec:setup}--\ref{sec:discussion}).
\end{itemize}

\section{Background}
\label{sec:background}

\subsection{Notation}
We write $\mathrm{op}(A)\in\mathbb{R}^{m\times k}$ and
$\mathrm{op}(B)\in\mathbb{R}^{k\times n}$, stored in column-major order, and
$u=2^{-53}$. For an integer $x$ and a modulus $m$, $\sym_m(x)$ denotes the
representative of $x \bmod m$ in a fixed window of $m$ consecutive integers that
contains $0$; for $m\le 256$ it fits in a signed byte. $\rint$ denotes rounding to
nearest with ties to even.

\subsection{Ozaki scheme I}
Ozaki scheme~I \cite{ozaki2012} writes each scaled row of $A$ and column of $B$ as a
sum of $L$ low-precision slices, $A\approx\sum_i A_i$ and $B\approx\sum_j B_j$,
chosen so that a low-precision engine computes every product $A_iB_j$ without
rounding. The approximation $AB\approx\sum_{i+j\le L+1}A_iB_j$ then requires
$L(L+1)/2$ products. The INT8 variants \cite{ootomo2024dgemm,uchino2025ozaki1int8}
use slices of about 7 bits.

\subsection{Ozaki scheme II}
\label{sec:oz2}
Ozaki scheme~II \cite{ozaki2025scheme2} consists of four steps. We follow the
formulation of the original paper, with the modifications described in
Section~\ref{sec:method}.

\emph{1. Integerization with shared exponents.} Each row $i$ of $A$ and each column
$j$ of $B$ is scaled by a power of two such that
\begin{equation}
  \hat a_{ip}=\rint(a_{ip}2^{D_i}),\qquad \hat b_{pj}=\rint(b_{pj}2^{E_j}),\qquad
  |\hat a_{ip}|,|\hat b_{pj}|\le 2^{t}.
  \label{eq:scale}
\end{equation}
Then $\hat c_{ij}=\sum_p\hat a_{ip}\hat b_{pj}\in\mathbb{Z}$ and
$c_{ij}\approx 2^{-D_i-E_j}\hat c_{ij}$. Errors arise only from the rounding in
Eq.~\eqref{eq:scale} and from the final rounding of $\hat c_{ij}2^{-D_i-E_j}$.

\emph{2. Residues.} For pairwise coprime moduli $m_1,\dots,m_s\le 256$ with
$M=\prod_t m_t$, the residue matrices $A^{(t)}=\sym_{m_t}(\hat A)$ and
$B^{(t)}=\sym_{m_t}(\hat B)$ have entries in $[-128,128]$ and are stored as INT8.

\emph{3. Exact INT8 products.} Each product $C^{(t)}=A^{(t)}B^{(t)}$ is computed with
INT32 accumulation and is exact as long as $k\cdot 128^2<2^{31}$, i.e.\
$k\le 131071$.

\emph{4. CRT reconstruction.} By the Chinese remainder theorem (see, e.g.,
\cite[\S4.3.2]{knuth1997taocp2}), with $M_t=M/m_t$ and $y_t=M_t^{-1}\bmod m_t$,
\begin{equation}
  \hat c_{ij}\equiv \sum_{t=1}^{s} C^{(t)}_{ij}\,M_ty_t \pmod M,
  \label{eq:crt}
\end{equation}
and $\hat c_{ij}$ is the unique representative with $|\hat c_{ij}|<M/2$, provided
that $|\hat c_{ij}|<M/2$ is known a priori. The method needs $s$ INT8 products, and
each modulus contributes $\log_2 m_t\approx 7.9$ bits to $\log_2 M$.

\subsection{Related implementations}
GEMMul8 \cite{gemmul8} implements Ozaki scheme~II for NVIDIA and AMD GPUs, and
ozIMMU \cite{ozimmu} implements the INT8 variant of Ozaki scheme~I. Kouya
\cite{kouya2026ozaki2cpu} implements Ozaki scheme~II for multiple precision on
Intel AMX and Arm SVE2. Uchino et al.\ \cite{uchino2025complex} extend Ozaki
scheme~II to CGEMM and ZGEMM with a Karatsuba (3M) formulation, also released in
GEMMul8, and report speed-ups on NVIDIA and AMD GPUs. Among vendor libraries,
cuBLAS has offered opt-in fixed-point emulation of FP64 matrix products (DGEMM and
ZGEMM) on INT8 Tensor Cores since CUDA~13.0 Update~2
\cite{cuda1302_relnotes,brower2025cublasemu}, whereas the oneMKL 2026 release
notes list no comparable feature \cite{onemkl_relnotes2026}. We are not aware of
any implementation for Intel GPUs.

\section{Target platform: Intel Arc Pro B60 (Xe2)}
\label{sec:platform}

The B60 has 20 Xe2 cores with 160 XMX matrix engines and 24\,GB of GDDR6 on a
192-bit bus with a nominal bandwidth of 456\,GB/s. Intel specifies 197\,TOP/s of
dense INT8, 12.28\,TFLOP/s of FP32 and a total board power of 120--200\,W,
depending on the partner design \cite{intel_arc_b60}; the maximum GPU clock of our
board is 2.4\,GHz. Table~\ref{tab:engines} lists the measured throughput of the
operations relevant here. Three properties of the device shape the
implementation:
\begin{enumerate}
\item[(i)] INT8 GEMM in oneMKL reaches 184--186\,TOP/s, about 250 times the native
  DGEMM rate of 0.76\,TFLOP/s.
\item[(ii)] FP64 vector arithmetic runs at $1/16$ of the FP32 rate, and the
  conversion from FP64 to INT64 is not native. Kernels that are meant to run at
  memory speed must therefore avoid FP64 altogether.
\item[(iii)] The memory bandwidth is $\approx$400\,GB/s for incompressible data,
  but the hardware compresses regular data (constant buffers are read at
  1.1\,TB/s). All memory-bound measurements therefore use random data.
\end{enumerate}

These numbers also settle the choice of backend. Following the per-bit argument of
\cite[\S3.3]{kouya2026ozaki2cpu}, one INT8 product contributes $\approx$7.85 bits
to $\log_2M$ at 186\,TOP/s, whereas one DGEMM in an FP64-based Ozaki-II backend
contributes $\approx$21.5 bits ($k=4096$) at 0.76\,TFLOP/s. Per bit of accuracy,
the INT8 path is thus $\approx$90 times cheaper on this device, and we do not
consider an FP64 backend.

\begin{table}[tbp]
\centering
\caption{Measured throughput on the test system (oneMKL 2026.1, random data).
Energy per operation is taken from the board energy counter.}
\label{tab:engines}
\small
\begin{tabular}{llrr}
\toprule
Operation & Size & Rate & pJ/op \\
\midrule
Memory copy (random data) & 4GiB & 397.20 GB/s & -- \\
Memory copy (constant data) & 4GiB & 1088.00 GB/s & -- \\
DGEMM, oneMKL, GPU & 4096 & 0.75 TFLOP/s & 98.0 \\
SGEMM, oneMKL, GPU & 4096 & 12.00 TFLOP/s & 10.1 \\
INT8 GEMM, oneMKL, GPU & 4096 & 170.90 TOP/s & 1.0 \\
DGEMM, oneMKL, GPU & 8192 & 0.76 TFLOP/s & 99.7 \\
SGEMM, oneMKL, GPU & 8192 & 11.89 TFLOP/s & 9.2 \\
INT8 GEMM, oneMKL, GPU & 8192 & 184.30 TOP/s & 1.0 \\
INT8 GEMM, batch of 16, GPU & 8192 & 185.50 TOP/s & 1.0 \\
\bottomrule
\end{tabular} {}
\end{table}

\section{The emulation algorithm as implemented}
\label{sec:method}

Algorithm~\ref{alg:gemm} summarizes the pipeline for one output tile; in the
common case the whole matrix $C$ forms a single tile. All steps run on the GPU, and
the host only computes tables of size $O(s)$. In the following we describe each
step and indicate where it follows the literature and where it departs from it.

\begin{algorithm}[tbp]
\caption{Emulated DGEMM (Ozaki scheme II), one tile, $\mathrm{op}=\mathtt{N}$}
\label{alg:gemm}
\begin{algorithmic}[1]
\Require $A\in\mathbb{R}^{m\times k}$, $B\in\mathbb{R}^{k\times n}$, $s$ moduli $m_1..m_s$
\State $t \gets$ largest integer with $k\,2^{2t+2}\le M$, capped at 100 \Comment{host, Eq.~\eqref{eq:tbound}}
\State $D_i,E_j\gets$ Eq.~\eqref{eq:exps} \Comment{kernels \textsc{RowMax}, \textsc{ColMax}}
\State $(A^{(t)})_{t=1}^s \gets \textsc{Split}(A, D)$,\ $(B^{(t)})_{t=1}^s \gets \textsc{Split}(B, E)$ \Comment{one pass each, Alg.~\ref{alg:split}}
\State $C^{(t)} \gets A^{(t)}B^{(t)}$ for $t=1..s$ \Comment{one strided-batch INT8 GEMM}
\State $C \gets \alpha\,\textsc{CRT}(C^{(1..s)}, D, E)+\beta C$ \Comment{Alg.~\ref{alg:crt}}
\end{algorithmic}
\end{algorithm}

\subsection{Choice of the scaling width $t$}
For $x>0$ let $\mathrm{e}(x)$ denote the binary exponent returned by the C function
\texttt{frexp}, i.e.\ the integer with $2^{\mathrm{e}(x)-1}\le x<2^{\mathrm{e}(x)}$. The
scaling exponents in Eq.~\eqref{eq:scale} are
\begin{equation}
  D_i=t-\mathrm{e}\big(\max_p|a_{ip}|\big),\qquad E_j=t-\mathrm{e}\big(\max_p|b_{pj}|\big),
  \label{eq:exps}
\end{equation}
so that $|a_{ip}2^{D_i}|<2^t$ and, after rounding, $|\hat a_{ip}|\le 2^t$ (likewise
for $B$). To choose $t$, we impose the a-priori bound
\begin{equation}
  |\hat c_{ij}|\le k\,2^{2t}\le M/4,
  \label{eq:tbound}
\end{equation}
which is half of the range required for uniqueness in Eq.~\eqref{eq:crt}. This factor
of $1/2$ is the price of the fast quotient evaluation described in
Section~\ref{sec:crt}; it costs one bit of $t$ per two moduli. Table~\ref{tab:moduli} lists $t$ for
typical values of $k$. For $s\ge 16$ and $k\le 2^{14}$ we have $t\ge 53$, so every
element within $t-53$ binades of its row maximum is represented exactly.

\begin{table}[tbp]
\centering
\caption{Moduli configuration: $\log_2 M$, number of 32-bit limbs used in the
reconstruction, and scaling width $t$ for several values of $k$.}
\label{tab:moduli}
\small
\begin{tabular}{rrrrrrr}
\toprule
$s$ & $\log_2 M$ & $N_L$ & $t$, $k{=}1024$ & $t$, $k{=}4096$ & $t$, $k{=}16384$ & $t$, $k{=}65536$ \\
\midrule
4 & 31.9 & 2 & 9 & 8 & 7 & 6 \\
6 & 47.8 & 2 & 17 & 16 & 15 & 14 \\
8 & 63.6 & 2 & 25 & 24 & 23 & 22 \\
10 & 79.3 & 3 & 33 & 32 & 31 & 30 \\
12 & 94.9 & 3 & 41 & 40 & 39 & 38 \\
14 & 110.4 & 4 & 49 & 48 & 47 & 46 \\
16 & 125.6 & 4 & 56 & 55 & 54 & 53 \\
18 & 140.8 & 5 & 64 & 63 & 62 & 61 \\
20 & 155.8 & 5 & 71 & 70 & 69 & 68 \\
22 & 170.6 & 6 & 79 & 78 & 77 & 76 \\
24 & 185.2 & 6 & 86 & 85 & 84 & 83 \\
\bottomrule
\end{tabular} {}
\end{table}

\subsection{Residue generation without floating-point arithmetic}
\label{sec:split}
Because FP64 is slow on Xe2, we compute $\hat a=\rint(a\,2^{D})$ directly from the
bit pattern of $a$. Writing $a=\pm\mu\,2^{e_a-1075}$ with integer significand
$\mu<2^{53}$, we obtain $\hat a=\pm\mu\,2^{\sigma}$ for
$\sigma=e_a-1075+D\ge 0$, and otherwise a right shift with round-to-nearest-even.
To avoid branching on the sign, we work with the non-negative offset value
$u=\hat a+2^{62}$ (or $\hat a+2^{104}$ when $t>62$), split into 21-bit limbs
$u=\sum_\ell u_\ell 2^{21\ell}$. For each modulus we form
\begin{equation}
  v=\sum_\ell u_\ell\,(2^{21\ell}\bmod m) + K_m < 2^{32},\quad
  K_m=(h_m-2^{62})\bmod m,\quad h_m=\lfloor (m-1)/2\rfloor,
\end{equation}
so that $(v\bmod m)-h_m\equiv\hat a\pmod m$ lies in the window
$[-h_m,\,m-1-h_m]\subset[-127,128]$. The value 128, which occurs only for
$m=256$, wraps to $-128$ in INT8; this is congruent and preserves the product
bound $128^2$. The reduction $v\bmod m$ uses Barrett's method \cite{barrett1986}
with a single \texttt{mul\_hi} and one conditional subtraction
($\mu_m=\lfloor 2^{32}/m\rfloor$, valid because $v<2^{32}$).
Algorithm~\ref{alg:split} shows the kernel. Each work-item converts four
consecutive elements of one column and writes one 32-bit word per modulus, so the
kernel reads $A$ once and writes $s$ bytes per element.

\begin{algorithm}[tbp]
\caption{\textsc{Split} kernel (one work-item, four rows $r..r{+}3$ of column $c$)}
\label{alg:split}
\begin{algorithmic}[1]
\For{$e=0..3$}
  \State $u \gets \textsc{ScaledOffsetInt}(X_{r+e,c},\ D_{r+e})$ \Comment{integer shifts with RNE, $u=\hat x+2^{62}$}
  \State $(u_0,u_1,u_2)\gets$ 21-bit limbs of $u$
\EndFor
\For{$t=1..s$} \Comment{unrolled; $s$ is a template parameter}
  \For{$e=0..3$}
    \State $v\gets u_0+u_1c_{t,1}+u_2c_{t,2}+K_t$ \Comment{$c_{t,\ell}=2^{21\ell}\bmod m_t$}
    \State $q\gets\texttt{mul\_hi}(v,\mu_t)$;\ $\rho\gets v-q\,m_t$;\ $\rho\gets\min(\rho,\rho-m_t)$ \Comment{unsigned}
    \State $\textit{byte}_e\gets(\rho-h_t)\ \&\ \mathtt{0xFF}$
  \EndFor
  \State store $(\textit{byte}_0..\textit{byte}_3)$ as one 32-bit word into plane $t$
\EndFor
\end{algorithmic}
\end{algorithm}

\subsection{INT8 products}
The $s$ products are issued as a single strided-batch call
(\texttt{oneapi::mkl::blas::gemm\_batch}, \texttt{int8}$\times$\texttt{int8}$\to$\texttt{int32}).
For $k>131071$ we split the $k$ dimension into chunks and sum their residues in the
reconstruction.

\subsection{CRT reconstruction with an FP32 quotient}
\label{sec:crt}
Let $r_t\in[0,2m_t)$ be any representative of $C^{(t)}_{ij}\bmod m_t$ and
$W_t=M_ty_t$. Then $Z=\sum_t r_tW_t\equiv\hat c_{ij}\pmod M$ and, exactly,
\begin{equation}
  \frac{Z}{M}=\sum_{t=1}^s\frac{r_t\,y_t}{m_t},
\end{equation}
an identity also used in \cite[\S4.4]{kouya2026ozaki2cpu}. Writing
$Z=\hat c_{ij}+qM$ with $q\in\mathbb{Z}$, the bound of Eq.~\eqref{eq:tbound} gives
$|Z/M-q|\le 1/4$. The integer $q$ can therefore be recovered from a crude
approximation of $Z/M$.

\begin{lemma}
Let $\tilde z$ be $\sum_t r_t\,\fl_{32}(y_t/m_t)$ accumulated with $s$ FP32 fused
multiply--adds. For $s\le 24$ and $m_t\le 256$, $|\tilde z-Z/M|<2^{-6}$; hence
$q=\rint(\tilde z)$.
\end{lemma}
\begin{proof}[Proof sketch]
All partial sums are below $2\cdot 256\cdot 24<2^{14}$, so each FMA contributes a
rounding error of at most $2^{-11}$, and each representation error satisfies
$|r_t(\fl_{32}(y_t/m_t)-y_t/m_t)|\le 2^{9}2^{-24}$. The total is below
$24(2^{-11}+2^{-15})<2^{-6}$, well inside the margin of $1/4$.
\end{proof}

Once $q$ is known, $\hat c_{ij}=Z-qM$ is an integer with
$|\hat c_{ij}|\le M/4<2^{32N_L-1}$, where $N_L=\lfloor\log_2M/32\rfloor+1$, and we
compute it exactly in two's complement modulo $2^{32N_L}$. Neither long division
nor correction steps are needed, in contrast with the FP64 quotient estimate and
subsequent integer correction of \cite{ozaki2025scheme2,kouya2026ozaki2cpu}. We
accumulate $Z$ in 16-bit limbs held in 32-bit registers. Every term satisfies
$r_t\cdot W_{t,\ell}<2^{9}\cdot2^{16}$, and 24 such terms stay below $2^{30}$, so
carries need to be propagated only once at the end; moreover, all multiplications
are $16\times16$-bit, which Xe2 executes natively, unlike $32\times32\to64$-bit
products. Finally, the $32N_L$-bit integer is normalized, its top 64 bits are
OR-ed with a sticky bit of the remaining bits and converted with
round-to-nearest-even, and the result is multiplied by the exact power of two
$2^{-D_i-E_j}$. The output is therefore the correctly rounded value of
$2^{-D_i-E_j}\hat c_{ij}$ unless it is subnormal.

\begin{algorithm}[tbp]
\caption{\textsc{CRT} kernel (one work-item per entry $(i,j)$)}
\label{alg:crt}
\begin{algorithmic}[1]
\State $\tilde z\gets 0$;\ $\mathit{acc}_\ell\gets 0$ for $\ell=0..2N_L-1$
\For{$t=1..s$} \Comment{unrolled}
  \State $r\gets \textsc{Barrett}(C^{(t)}_{ij}+2^{31})+K'_t$ \Comment{$r\in[0,2m_t)$, $K'_t=(-2^{31})\bmod m_t$}
  \State $\tilde z\gets\mathrm{fma}_{32}(r,\ y_t/m_t,\ \tilde z)$
  \State $\mathit{acc}_\ell\gets\mathit{acc}_\ell+r\cdot W_{t,\ell}$ for all $\ell$ \Comment{16-bit limbs of $W_t=M_ty_t$}
\EndFor
\State $q\gets\rint(\tilde z)$;\ $\mathit{acc}_\ell\gets\mathit{acc}_\ell-q\,M_\ell$; propagate carries \Comment{$\hat c=Z-qM$ exactly}
\State $c\gets\textsc{RoundToDouble}(\hat c)\cdot2^{-D_i-E_j}$ \Comment{normalize, sticky bit, RNE}
\State $C_{ij}\gets\alpha c+\beta C_{ij}$
\end{algorithmic}
\end{algorithm}

\subsection{Error bound and the precision knob}
Combining the rounding in Eq.~\eqref{eq:scale} with the final rounding gives, for every
entry,
\begin{equation}
  |\tilde c_{ij}-c_{ij}|\le\tfrac12 2^{-D_i}\sum_p|b_{pj}|+\tfrac12 2^{-E_j}\sum_p|a_{ip}|
  +\tfrac14k\,2^{-D_i-E_j}+u|\tilde c_{ij}|,
  \label{eq:bound}
\end{equation}
which our tests verify element by element (Section~\ref{sec:correctness}). Since
$2^{-D_i}\approx2^{-t}\max_p|a_{ip}|$, the first three terms are of order
$2^{-t}(|A||B|)_{ij}$. The number of moduli $s$ thus determines $t$ and, through
it, the accuracy. The helper \texttt{moduli\_for\_bits}$(k,b)$ returns the
smallest $s$ for which $t\ge b$.

\section{Complex matrix multiplication (ZGEMM)}
\label{sec:zgemm}
For $A=A_r+iA_i$ and $B=B_r+iB_i$ we use one shared exponent per row of
$\mathrm{op}(A)$, taken over both real and imaginary parts, and one per column of
$\mathrm{op}(B)$. The integerized parts then satisfy
$|\mathrm{Re}\,\hat C|,|\mathrm{Im}\,\hat C|\le2k\,2^{2t}$, so $t$ is chosen
with $2k$ in place of $k$ in Eq.~\eqref{eq:tbound}. As in \cite{uchino2025complex},
we use the 3M (Karatsuba) form
\begin{equation}
  P_0=\hat A_r\hat B_r,\quad P_1=\hat A_i\hat B_i,\quad
  P_2=(\hat A_r+\hat A_i)(\hat B_r+\hat B_i),\qquad
  \mathrm{Re}\,\hat C=P_0-P_1,\quad \mathrm{Im}\,\hat C=P_2-P_0-P_1,
\end{equation}
evaluated entirely in residues. The split kernel emits three groups of $s$ planes,
$\sym(\hat a_r)$, $\sym(\hat a_i)$ and $\sym(\sym(\hat a_r)+\sym(\hat a_i))$, and
implements conjugation by negating the imaginary residue. Three strided-batch INT8
GEMMs follow, and the CRT kernel forms $(P_0-P_1)\bmod m_t$ and
$(P_2-P_0-P_1)\bmod m_t$ before two independent reconstructions. Because every
product is exact, the 3M form incurs no accuracy penalty here, unlike its
floating-point counterpart \cite{higham1992complex3m}, and it saves a quarter of
the INT8 GEMMs ($3s$ instead of $4s$). Algorithm~\ref{alg:zgemm} summarizes the
procedure.

\begin{algorithm}[tbp]
\caption{Emulated ZGEMM, 3M form}
\label{alg:zgemm}
\begin{algorithmic}[1]
\State $t\gets$ Eq.~\eqref{eq:tbound} with $2k$;\ $D_i,E_j$ from $\max(|\mathrm{Re}|,|\mathrm{Im}|)$
\State $(A_r^{(t)},A_i^{(t)},A_\Sigma^{(t)})\gets\textsc{SplitC}(A,D)$; same for $B$ \Comment{$A_\Sigma^{(t)}=\sym(A_r^{(t)}+A_i^{(t)})$}
\State $P_0^{(t)}\gets A_r^{(t)}B_r^{(t)}$,\ $P_1^{(t)}\gets A_i^{(t)}B_i^{(t)}$,\ $P_2^{(t)}\gets A_\Sigma^{(t)}B_\Sigma^{(t)}$ \Comment{3 batched calls}
\State $\mathrm{Re}\gets\textsc{CRT}\big((P_0-P_1)\bmod m_t\big)$,\ $\mathrm{Im}\gets\textsc{CRT}\big((P_2-P_0-P_1)\bmod m_t\big)$
\State $C\gets\alpha(\mathrm{Re}+i\,\mathrm{Im})2^{-D_i-E_j}+\beta C$
\end{algorithmic}
\end{algorithm}

\section{Implementation decisions specific to Intel Arc (Xe2)}
\label{sec:arc}

\paragraph{Kernel code generation.} The split and CRT kernels are templates over
$s$. All loops over moduli are unrolled, and the modulus tables (at most 1\,KB) are
passed as kernel arguments instead of being read from memory. Both kernels use
sub-groups of 32 work-items; the complex split uses 16 because of register
pressure. With integer-only residue generation, 16-bit limbs and a $[0,2m)$
representative, both kernels are memory-bound: at $N=8192$ the split kernel
streams 270--340\,GB/s and the CRT kernel 373--416\,GB/s, i.e.\ close to the
measured memory bandwidth (Table~\ref{tab:kernels}).

\paragraph{Using the vendor INT8 GEMM.} The INT8 products are computed by the
batched GEMM of oneMKL (\texttt{gemm\_batch}), which reaches 184--193\,TOP/s on the B60. A
hand-written \texttt{joint\_matrix} kernel is possible, but offers no advantage: the
INT8 GEMM is currently a core operation of the vendor's inference stack and is
tuned by the vendor for each generation of XMX engines. Relying on it lets the
emulation inherit these improvements without changes to our code, and confines our own kernels to the
memory-bound split and CRT phases. We do, however, adapt the calling pattern to
three behaviours of the routine that we measured:
(i) with \texttt{transb=trans} its throughput drops from 175 to 33\,TOP/s, so we
first transpose a transposed $\mathrm{op}(B)$ through shared local memory into a
temporary tile and always store the residue planes of $B$ in \texttt{N} layout;
(ii) a leading dimension of the $A$ planes that is a multiple of
$12288=3\cdot4096$ reduces the throughput to 120--145\,TOP/s, which a padding of 64
bytes eliminates (Table~\ref{tab:lda});
(iii) tiles whose row count is not a power of two run 5--17\% slower, so the tiler
prefers $m_b=2^j$ (Figure~\ref{fig:shapes}).

\begin{table}[tbp]
\centering
\begin{minipage}[t]{0.40\textwidth}
\centering
\raggedright
\caption{Split and CRT kernels ($N=8192$): time (ms) and effective bandwidth
(GB/s).}
\label{tab:kernels}
\centering
\footnotesize
\begin{tabular}{lrrr}
\toprule
kernel & $s$ & ms & GB/s \\
\midrule
split & 8 & 3.18 & 337 \\
split & 16 & 4.73 & 340 \\
split & 24 & 7.94 & 270 \\
\midrule
CRT & 8 & 7.20 & 373 \\
CRT & 16 & 11.90 & 406 \\
CRT & 24 & 16.78 & 416 \\
\bottomrule
\end{tabular} {}
\end{minipage}\hfill
\begin{minipage}[t]{0.55\textwidth}
\centering
\raggedright
\caption{oneMKL INT8 GEMM throughput (TOP/s), $k=8192$, with the natural leading
dimension (lda$=m$) and with a padded one.}
\label{tab:lda}
\centering
\footnotesize
\begin{tabular}{rrrrr}
\toprule
$m$ & $n$ & lda$=m$ & lda & padded \\
\midrule
4096 & 4096 & 183 & 4160 & 182 \\
4096 & 16384 & 185 & 4160 & 185 \\
8192 & 4096 & 184 & 8256 & 184 \\
8192 & 16384 & 185 & 8256 & 185 \\
12288 & 4096 & 129 & 12352 & 175 \\
12288 & 16384 & 122 & 12352 & 181 \\
16384 & 4096 & 184 & 16448 & 185 \\
16384 & 16384 & 185 & 16448 & 186 \\
\bottomrule
\end{tabular} {}
\end{minipage}
\end{table}

\begin{figure}[tbp]
\centering
\includegraphics[width=0.55\textwidth]{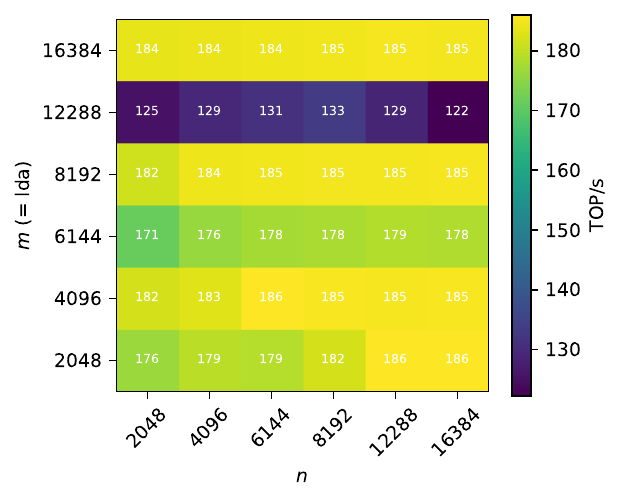}
\caption{oneMKL INT8 \texttt{gemm\_batch} throughput (TOP/s) by tile shape, $k=8192$.}
\label{fig:shapes}
\end{figure}

\paragraph{Memory-aware tiling.} A tile of size $m_b\times n_b$ requires a workspace
of $s(m_bk+kn_b+4m_bn_b)$ bytes ($\times3$ for ZGEMM). When this exceeds the
memory budget, $C$ is tiled. Candidate tile sizes are scored with a simple cost
model, consisting of the GEMM time with a shape-efficiency factor plus the split
and CRT traffic, and the cheapest feasible tiling is selected. The budget is
derived from the driver's free-memory query. We observed, however, that the Linux
\texttt{xe} driver permits VRAM over-subscription by evicting buffers to system
memory. After such an eviction the free-memory query overstates the available VRAM,
and a workspace that seemingly fits degrades performance by an order of magnitude
while consuming host RAM. The API therefore accepts a hint of the caller's
resident memory, which caps the budget.

\paragraph{Launch and compilation.} With 23 values of $s$ and several kernel
variants, the module contains about 150 kernels. When built just-in-time without
per-kernel device-code splitting, the first call compiles all of them, which takes
38\,s. With \texttt{-fsycl-device-code-split=per\_kernel} or ahead-of-time
compilation for \texttt{bmg-g21}, the first call takes 30--60\,ms
(Table~\ref{tab:firstcall}). All work is submitted asynchronously to one in-order
queue, and \texttt{gemm} and \texttt{zgemm} return a \texttt{sycl::event}.

\begin{table}[tbp]
\centering
\caption{First-call latency in a fresh process (s).}
\label{tab:firstcall}
\footnotesize
\begin{tabular}{lrrrr}
\toprule
build & $s{=}8$ (first) & $s{=}12$ & $s{=}16$ & warm call \\
\midrule
AOT (bmg-g21) & 0.029 & 0.002 & 0.003 & 0.0004 \\
JIT, per-kernel split & 0.057 & 0.004 & 0.004 & 0.0004 \\
JIT, default & 38.379 & 0.001 & 0.001 & 0.0005 \\
\bottomrule
\end{tabular} {}
\end{table}

\paragraph{What did not help.} Running the memory-bound CRT of one tile on a second
queue, concurrently with the INT8 GEMM of the next tile, produced no overlap,
because the GEMM occupies all Xe cores. Two-pass row-maximum reductions and vector
loads did not improve the scaling kernels, which already run at memory speed.

\section{Experimental setup}
\label{sec:setup}
All measurements were taken on an Intel Arc Pro B60 (24\,GB, maximum GPU clock
2400\,MHz) under Ubuntu 26.04.1 LTS with Linux 7.0 and the \texttt{xe} kernel driver. The
display is driven by a separate GPU, so the B60 carries no graphics workload.
Because the library is fully device-resident, the host system does not affect the
GPU measurements; the host CPU serves only to compute reference results with
Intel MKL DGEMM/ZGEMM for the accuracy comparison. We used the Intel oneAPI DPC++/C++ Compiler 2026.1.1
(SYCL 2020 \cite{sycl2020}), oneMKL 2026.1 \cite{onemkl} and the Level Zero runtime,
with ahead-of-time compilation for \texttt{bmg-g21}.

Matrices are generated on the GPU. Times are wall-clock times from submission to
completion of \texttt{gemm}, with the workspace allocated by a previous call; we
report the best of 3--7 repetitions (medians differ by less than 1\%). Emulated
FP64 rates count $2mnk$ flops for DGEMM and $8mnk$ for ZGEMM. Energy is read from
the board energy counter of the \texttt{xe} driver over runs of at least 2\,s.

We use the following test matrices: $\mathcal{U}[-1,1)$ (as in
\cite{kouya2026ozaki2cpu}); $(\mathcal{U}(0,1)-\frac12)\exp(\phi\,\mathcal{N}(0,1))$
with $\phi\in\{0.5,1,2,4\}$ (as in \cite{ozaki2025scheme2}); $\sin/\cos$ matrices;
and two adversarial families described in Section~\ref{sec:worst}. Accuracy is
measured against a double-double GEMM computed on the GPU (error $\approx2^{-100}$
relative to $|A||B|$) on $512\times512$ outputs. We report the maximum error in
ulps of the exact result, the fraction of correctly rounded entries, the
componentwise error $\max|\tilde c_{ij}-c_{ij}|/(|A||B|)_{ij}$ and the normwise
relative error. All experiments can be reproduced (see Code and data availability).

\section{Results}
\label{sec:results}

\subsection{Correctness}
\label{sec:correctness}
The test program (108 tests) checks the bound of Eq.~\eqref{eq:bound} element-wise against
the double-double reference for all combinations of \texttt{N/T/C}, $\alpha$ and
$\beta$, leading dimensions larger than the row count, $s=2..24$, $k$ from 1 to
300\,000, and tiling under a workspace limit of 4--8\,MB; tiled results are bitwise
identical to untiled ones. Two unit tests follow \cite{kouya2026ozaki2cpu}. The
split kernel is compared with \texttt{\_\_int128} arithmetic on 2.2 million
residues, including subnormals and signed zeros, and the CRT kernel reconstructs
8192 random exact integers with $|\hat c|\le M/4$ (including $|\hat c|=M/4$) for
every $s=2..24$. Both tests report zero mismatches. Integer inputs give bitwise
exact results, and Inf/NaN inputs propagate NaN to the affected row or column.

\subsection{Accuracy}
\label{sec:accuracy}
Figure~\ref{fig:accuracy} and Table~\ref{tab:accuracy} show how the error depends on
$s$. The componentwise error decreases by a factor of about $2^{4}$ per modulus,
as predicted by $t\approx(\log_2M-\log_2k)/2$. The componentwise accuracy of native
DGEMM is reached at $s=14$--$15$ for well-scaled data and at $s=17$ for $\phi=2$
(Table~\ref{tab:mins}). Beyond this point the result becomes correctly rounded. For
$\mathcal{U}[-1,1)$, all entries are correctly rounded from $s=15$--$16$ on,
whereas native DGEMM rounds only 1.7--3\% (oneMKL on the GPU) and 2--13\% (MKL on
the CPU) of the entries correctly, with maximum errors of $10^5$--$10^7$ ulps caused
by cancellation. For
inputs with wide exponent ranges ($\phi=0.5$--$2$ and $\sin/\cos$), correct
rounding requires $t>62$ and thus the extended-width path with $s=19$--$23$.

\begin{figure}[tbp]
\centering
\includegraphics[width=\textwidth]{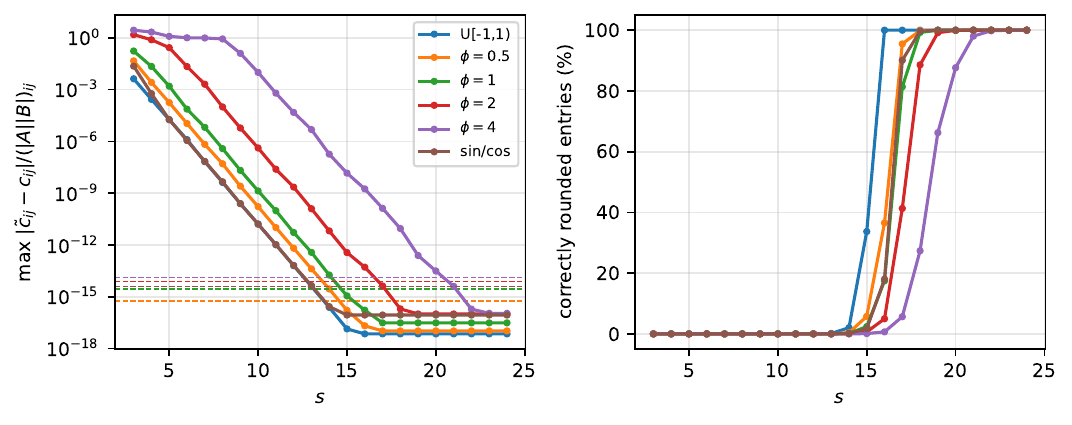}
\caption{Accuracy versus the number of moduli, $k=4096$. Left: componentwise error;
dashed lines show native DGEMM on the GPU. Right: fraction of correctly rounded
entries.}
\label{fig:accuracy}
\end{figure}

\begin{table}[tbp]
\centering
\caption{Accuracy against the double-double reference ($512\times512$, $k=4096$).}
\label{tab:accuracy}
\footnotesize
\begin{tabular}{llrrrrr}
\toprule
Data & Method & $t$ & max ulp & CR \% & max $\frac{|e|}{|A||B|}$ & $\frac{\|e\|_F}{\|C\|_F}$ \\
\midrule
U[-1,1) & DGEMM GPU & 53 & $3.0\cdot 10^{6}$ & 2.64 & $5.3\cdot 10^{-16}$ & $2.1\cdot 10^{-15}$ \\
U[-1,1) & DGEMM CPU & 53 & $5.0\cdot 10^{5}$ & 10.13 & $6.0\cdot 10^{-17}$ & $4.9\cdot 10^{-16}$ \\
U[-1,1) & Oz-II $s{=}12$ & 40 & $4.3\cdot 10^{8}$ & 0.01 & $6.3\cdot 10^{-14}$ & $6.4\cdot 10^{-13}$ \\
U[-1,1) & Oz-II $s{=}14$ & 48 & $2.1\cdot 10^{6}$ & 2.03 & $2.4\cdot 10^{-16}$ & $2.5\cdot 10^{-15}$ \\
U[-1,1) & Oz-II $s{=}16$ & 55 & 0.50 & 100.00 & $7.0\cdot 10^{-18}$ & $4.7\cdot 10^{-17}$ \\
U[-1,1) & Oz-II $s{=}18$ & 63 & 0.50 & 100.00 & $7.0\cdot 10^{-18}$ & $4.7\cdot 10^{-17}$ \\
U[-1,1) & Oz-II $s{=}20$ & 70 & 0.50 & 100.00 & $7.0\cdot 10^{-18}$ & $4.7\cdot 10^{-17}$ \\
U[-1,1) & Oz-II $s{=}24$ & 85 & 0.50 & 100.00 & $7.0\cdot 10^{-18}$ & $4.7\cdot 10^{-17}$ \\
\midrule
$\phi$=2 & DGEMM GPU & 53 & $5.5\cdot 10^{6}$ & 2.90 & $7.8\cdot 10^{-15}$ & $1.9\cdot 10^{-15}$ \\
$\phi$=2 & DGEMM CPU & 53 & $4.9\cdot 10^{6}$ & 11.50 & $1.7\cdot 10^{-15}$ & $5.0\cdot 10^{-16}$ \\
$\phi$=2 & Oz-II $s{=}12$ & 40 & $1.7\cdot 10^{11}$ & 0.00 & $2.2\cdot 10^{-9}$ & $3.1\cdot 10^{-11}$ \\
$\phi$=2 & Oz-II $s{=}14$ & 48 & $2.6\cdot 10^{9}$ & 0.04 & $6.5\cdot 10^{-12}$ & $1.2\cdot 10^{-13}$ \\
$\phi$=2 & Oz-II $s{=}16$ & 55 & $2.4\cdot 10^{6}$ & 4.98 & $5.1\cdot 10^{-14}$ & $9.2\cdot 10^{-16}$ \\
$\phi$=2 & Oz-II $s{=}18$ & 63 & 8871.00 & 88.58 & $2.0\cdot 10^{-16}$ & $4.3\cdot 10^{-17}$ \\
$\phi$=2 & Oz-II $s{=}20$ & 70 & 18.26 & 99.93 & $1.0\cdot 10^{-16}$ & $4.2\cdot 10^{-17}$ \\
$\phi$=2 & Oz-II $s{=}24$ & 85 & 0.50 & 100.00 & $1.0\cdot 10^{-16}$ & $4.2\cdot 10^{-17}$ \\
\midrule
sincos & DGEMM GPU & 53 & $1.7\cdot 10^{6}$ & 1.71 & $3.9\cdot 10^{-15}$ & $1.3\cdot 10^{-15}$ \\
sincos & DGEMM CPU & 53 & $2.3\cdot 10^{6}$ & 1.98 & $4.2\cdot 10^{-16}$ & $2.7\cdot 10^{-16}$ \\
sincos & Oz-II $s{=}12$ & 40 & $2.8\cdot 10^{10}$ & 0.00 & $6.5\cdot 10^{-14}$ & $4.1\cdot 10^{-13}$ \\
sincos & Oz-II $s{=}14$ & 48 & $6.1\cdot 10^{7}$ & 0.18 & $2.7\cdot 10^{-16}$ & $1.6\cdot 10^{-15}$ \\
sincos & Oz-II $s{=}16$ & 55 & $2.6\cdot 10^{5}$ & 18.01 & $8.7\cdot 10^{-17}$ & $3.9\cdot 10^{-17}$ \\
sincos & Oz-II $s{=}18$ & 63 & 58.52 & 99.70 & $8.7\cdot 10^{-17}$ & $3.9\cdot 10^{-17}$ \\
sincos & Oz-II $s{=}20$ & 70 & 0.50 & 100.00 & $8.7\cdot 10^{-17}$ & $3.9\cdot 10^{-17}$ \\
sincos & Oz-II $s{=}24$ & 85 & 0.50 & 100.00 & $8.7\cdot 10^{-17}$ & $3.9\cdot 10^{-17}$ \\
\bottomrule
\end{tabular} {}
\end{table}

\begin{table}[tbp]
\centering
\caption{Smallest $s$ that reaches (a) the componentwise accuracy of native GPU
DGEMM, (b) 99.9\% and (c) 100\% correctly rounded entries.}
\label{tab:mins}
\footnotesize
\begin{tabular}{lrrrrr}
\toprule
Data & $k$ & (a) $\approx$DGEMM & (b) 99.9\% CR & (c) 100\% CR & DGEMM GPU CR \% \\
\midrule
U[-1,1) & 256 & 14 & 15 & 15 & 10.6 \\
U[-1,1) & 4096 & 14 & 16 & 16 & 2.6 \\
U[-1,1) & 65536 & 14 & 16 & 16 & 0.6 \\
$\phi$=0.5 & 256 & 15 & 18 & 19 & 10.5 \\
$\phi$=0.5 & 4096 & 15 & 19 & 20 & 2.6 \\
$\phi$=0.5 & 65536 & 15 & 19 & 22 & 0.7 \\
$\phi$=1 & 256 & 15 & 18 & 20 & 11.0 \\
$\phi$=1 & 4096 & 15 & 19 & 21 & 2.7 \\
$\phi$=1 & 65536 & 16 & 20 & 22 & 0.7 \\
$\phi$=2 & 256 & 17 & 19 & 21 & 12.2 \\
$\phi$=2 & 4096 & 17 & 20 & 22 & 2.9 \\
$\phi$=2 & 65536 & 17 & 21 & 23 & 0.7 \\
$\phi$=4 & 256 & 21 & 22 & 24 & 13.7 \\
$\phi$=4 & 4096 & 21 & 23 & -- & 3.3 \\
$\phi$=4 & 65536 & 21 & 24 & -- & 0.8 \\
sincos & 256 & 14 & 18 & 19 & 7.7 \\
sincos & 4096 & 14 & 19 & 20 & 1.7 \\
sincos & 65536 & 13 & 20 & 21 & 0.4 \\
\bottomrule
\end{tabular} {}
\end{table}

\subsection{Worst cases}
\label{sec:worst}
Table~\ref{tab:worst} collects inputs designed to stress the method.

\emph{Boundary of the CRT range.} If all entries equal $1-2^{-53}$, they round to
$\hat a=2^t$ and give $|\hat c|=k2^{2t}=M/4$, the limit of Eq.~\eqref{eq:tbound}. The
results are exact for all $s$.

\emph{Spike.} One element per row of $A$ is $2^{40}$ times larger than the others.
The shared row exponent then leaves the remaining elements only $t-40$ bits, and at
$s=16$ the componentwise error is $10^{-10}$, four orders of magnitude worse than
DGEMM. Increasing $s$ to 20 restores DGEMM-level accuracy, and $s=24$ gives correct
rounding. This is the fundamental limitation of all shared-exponent schemes.

\emph{Cancellation.} With $A=[X\ X]$ and $B=[Y;\,-Y+2^{-30}Z]$ we obtain
$|C|\approx2^{-30}|A||B|$. Native DGEMM loses 30 bits (normwise relative error
$2\times10^{-6}$ on the GPU and $7\times10^{-7}$ on the CPU). Ozaki scheme~II with $s\ge20$, in contrast, represents the inputs
exactly, computes $\hat c$ exactly and rounds only once: the normwise error is
$9\times10^{-16}$ at $s=20$, and the result is correctly rounded at $s=24$.

\begin{table}[tbp]
\centering
\caption{Worst-case inputs ($512\times512$, $k=4096$).}
\label{tab:worst}
\footnotesize
\begin{tabular}{llrrrrr}
\toprule
Data & Method & $t$ & max ulp & CR \% & max $\frac{|e|}{|A||B|}$ & $\frac{\|e\|_F}{\|C\|_F}$ \\
\midrule
spike & DGEMM GPU & 53 & 86.57 & 3.57 & $1.8\cdot 10^{-14}$ & $2.2\cdot 10^{-15}$ \\
spike & DGEMM CPU & 53 & 16.76 & 15.54 & $3.1\cdot 10^{-15}$ & $5.0\cdot 10^{-16}$ \\
spike & Oz-II $s{=}12$ & 40 & $2.2\cdot 10^{10}$ & 0.00 & $3.9\cdot 10^{-6}$ & $2.2\cdot 10^{-11}$ \\
spike & Oz-II $s{=}16$ & 55 & $6.7\cdot 10^{5}$ & 8.05 & $1.3\cdot 10^{-10}$ & $6.8\cdot 10^{-16}$ \\
spike & Oz-II $s{=}20$ & 70 & 7.46 & 99.93 & $9.1\cdot 10^{-16}$ & $4.7\cdot 10^{-17}$ \\
spike & Oz-II $s{=}24$ & 85 & 0.50 & 100.00 & $1.1\cdot 10^{-16}$ & $4.7\cdot 10^{-17}$ \\
spike & Oz-I 36 GEMMs & 54 & $5.8\cdot 10^{6}$ & 1.30 & $1.0\cdot 10^{-9}$ & $4.1\cdot 10^{-15}$ \\
spike & Oz-I 45 GEMMs & 61 & 59480.00 & 66.32 & $1.1\cdot 10^{-11}$ & $6.0\cdot 10^{-17}$ \\
\midrule
cancel & DGEMM GPU & 53 & $5.6\cdot 10^{15}$ & 0.00 & $3.2\cdot 10^{-16}$ & $2.3\cdot 10^{-6}$ \\
cancel & DGEMM CPU & 53 & $8.5\cdot 10^{14}$ & 0.00 & $5.2\cdot 10^{-17}$ & $7.2\cdot 10^{-7}$ \\
cancel & Oz-II $s{=}12$ & 40 & $7.2\cdot 10^{17}$ & 0.00 & $4.1\cdot 10^{-14}$ & $6.9\cdot 10^{-4}$ \\
cancel & Oz-II $s{=}16$ & 55 & $2.1\cdot 10^{13}$ & 0.00 & $3.6\cdot 10^{-19}$ & $6.0\cdot 10^{-9}$ \\
cancel & Oz-II $s{=}20$ & 70 & $9.8\cdot 10^{5}$ & 99.42 & $4.2\cdot 10^{-25}$ & $9.4\cdot 10^{-16}$ \\
cancel & Oz-II $s{=}24$ & 85 & 0.50 & 100.00 & $5.5\cdot 10^{-27}$ & $4.7\cdot 10^{-17}$ \\
cancel & Oz-I 36 GEMMs & 54 & $3.2\cdot 10^{14}$ & 0.00 & $1.2\cdot 10^{-17}$ & $1.5\cdot 10^{-7}$ \\
cancel & Oz-I 45 GEMMs & 61 & $1.1\cdot 10^{12}$ & 0.00 & $8.2\cdot 10^{-20}$ & $1.2\cdot 10^{-9}$ \\
\bottomrule
\end{tabular} {}
\end{table}

\subsection{DGEMM performance}
Table~\ref{tab:perf} and Figure~\ref{fig:perf} report the emulated rate for square
problems. For $N\ge 2048$ the emulation outperforms native DGEMM for every $s$. At
$N=8192$ it reaches 17.3, 12.1 and 9.3\,TFLOP/s for $s=8$, 12 and 16, which is
22.8, 15.9 and 12.2 times the native DGEMM rate (0.76\,TFLOP/s). Problems up to $N=24576$ are tiled
automatically ($24576^3$ requires $\approx$9\,GB of workspace in addition to the
14.5\,GB occupied by $A$, $B$ and $C$) at an efficiency loss of 0--10\%. For
$N\le256$, kernel-launch overheads ($\approx$0.15\,ms) dominate, and native DGEMM is
faster. Non-square shapes (Table~\ref{tab:shapes}) behave as the operation count
predicts: a small inner dimension ($k=256$) leaves the memory-bound phases exposed
(2.4\,TFLOP/s at $s=8$), whereas a large one approaches the GEMM bound
(19.6\,TFLOP/s for $k=32768$). Transposed operands cost at most 7\%.

\begin{figure}[tbp]
\centering
\includegraphics[width=0.75\textwidth]{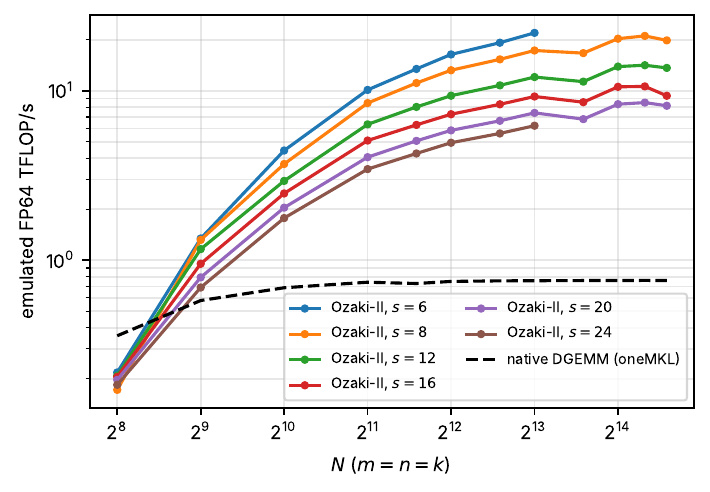}
\caption{Emulated FP64 rate for square DGEMM versus $N$.}
\label{fig:perf}
\end{figure}

\begin{table}[tbp]
\centering
\caption{Emulated DGEMM rate (TFLOP/s) for square matrices of order $N$.}
\label{tab:perf}
\footnotesize
\begin{tabular}{rrrrrrrr}
\toprule
$N$ & $s{=}4$ & $s{=}8$ & $s{=}12$ & $s{=}16$ & $s{=}20$ & $s{=}24$ & DGEMM GPU \\
\midrule
256 & 0.16 & 0.17 & 0.21 & 0.20 & 0.20 & 0.18 & 0.36 \\
512 & 1.28 & 1.32 & 1.16 & 0.95 & 0.79 & 0.69 & 0.58 \\
1024 & 4.82 & 3.70 & 2.94 & 2.48 & 2.04 & 1.77 & 0.69 \\
2048 & 9.33 & 8.47 & 6.34 & 5.09 & 4.06 & 3.45 & 0.74 \\
3072 & 17.11 & 11.12 & 8.04 & 6.30 & 5.07 & 4.27 & 0.73 \\
4096 & 21.14 & 13.23 & 9.36 & 7.28 & 5.84 & 4.94 & 0.75 \\
6144 & 25.84 & 15.34 & 10.77 & 8.33 & 6.66 & 5.61 & 0.76 \\
8192 & 29.90 & 17.30 & 12.06 & 9.26 & 7.41 & 6.23 & 0.76 \\
12288 & -- & 16.70 & 11.33 & 8.57 & 6.81 & -- & 0.76 \\
16384 & -- & 20.32 & 13.91 & 10.56 & 8.34 & -- & 0.76 \\
20480 & -- & 21.10 & 14.17 & 10.62 & 8.52 & -- & 0.76 \\
24576 & -- & 19.86 & 13.66 & 9.35 & 8.17 & -- & 0.76 \\
\bottomrule
\end{tabular} {}
\end{table}

\begin{table}[tbp]
\centering
\caption{Non-square shapes: DGEMM rate (TFLOP/s), native and emulated.}
\label{tab:shapes}
\footnotesize
\begin{tabular}{lrrr}
\toprule
$m\times k\times n$ & native DGEMM & $s{=}8$ & $s{=}16$ \\
\midrule
8192$\times$256$\times$8192 & 0.75 & 2.35 & 1.32 \\
8192$\times$1024$\times$8192 & 0.76 & 7.43 & 4.12 \\
8192$\times$32768$\times$8192 & 0.76 & 19.62 & 10.38 \\
8192$\times$65536$\times$8192 & 0.75 & 19.65 & 10.03 \\
65536$\times$4096$\times$1024 & 0.76 & 10.64 & 6.20 \\
1024$\times$4096$\times$65536 & 0.76 & 10.97 & 6.37 \\
4000$\times$8000$\times$4000 & 0.71 & 14.02 & 7.73 \\
16384$\times$2048$\times$16384 & 0.76 & 11.51 & 6.20 \\
\bottomrule
\end{tabular} {}
\end{table}

\paragraph{Phase breakdown.} Table~\ref{tab:phases} and Figure~\ref{fig:phases}
split the run time into scaling, splitting, INT8 GEMM and CRT. Unlike on the CPUs
studied in \cite{kouya2026ozaki2cpu}, where the residue conversion dominates, on the
B60 the INT8 GEMM accounts for 75--80\% of the time at $N=8192$ and 88--90\% at
$N=24576$, running at 186\,TOP/s, while the split and CRT kernels run at the memory
bandwidth. The performance model of Section~\ref{sec:model} reproduces the measured
times to within 10\% on average.

\begin{table}[tbp]
\centering
\caption{Phase breakdown (ms).}
\label{tab:phases}
\footnotesize
\begin{tabular}{rrrrrrrrl}
\toprule
$N$ & $s$ & scale & split & GEMM & TOP/s & CRT & GEMM \% & tiles \\
\midrule
2048 & 8 & 0.25 & 0.41 & 0.90 & 153 & 0.47 & 44 & 1x1 \\
2048 & 12 & 0.25 & 0.51 & 1.34 & 153 & 0.63 & 49 & 1x1 \\
2048 & 16 & 0.25 & 0.60 & 1.76 & 156 & 0.78 & 52 & 1x1 \\
2048 & 20 & 0.25 & 0.89 & 2.20 & 156 & 0.90 & 52 & 1x1 \\
2048 & 24 & 0.27 & 1.02 & 2.66 & 155 & 1.05 & 53 & 1x1 \\
\midrule
8192 & 8 & 2.64 & 6.29 & 47.31 & 186 & 7.21 & 75 & 1x1 \\
8192 & 12 & 2.64 & 7.86 & 70.97 & 186 & 9.52 & 78 & 1x1 \\
8192 & 16 & 2.68 & 9.40 & 94.63 & 186 & 11.92 & 80 & 1x1 \\
8192 & 20 & 2.66 & 13.71 & 118.09 & 186 & 13.77 & 80 & 1x1 \\
8192 & 24 & 2.70 & 15.83 & 141.78 & 186 & 16.23 & 80 & 1x1 \\
\midrule
24576 & 8 & 22.75 & 86.64 & 1332.98 & 178 & 65.33 & 88 & 3x2 \\
24576 & 12 & 22.73 & 108.26 & 1954.12 & 182 & 87.03 & 90 & 6x2 \\
24576 & 16 & 22.82 & 215.98 & 2836.96 & 167 & 108.34 & 89 & 3x4 \\
24576 & 20 & 22.82 & 248.67 & 3229.24 & 184 & 125.77 & 89 & 6x3 \\
\bottomrule
\end{tabular} {}
\end{table}

\begin{figure}[tbp]
\centering
\includegraphics[width=0.7\textwidth]{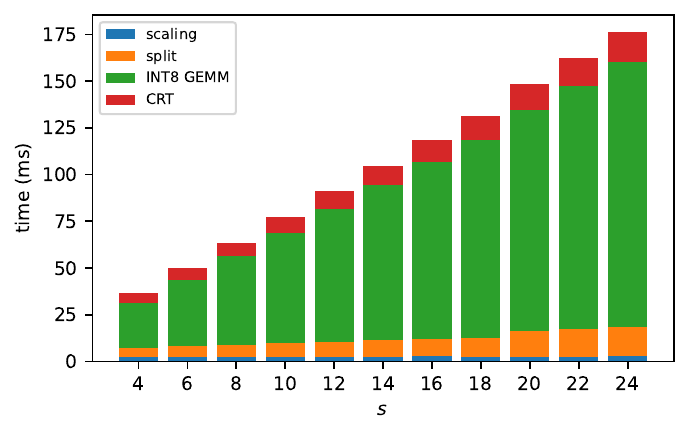}
\caption{Phase breakdown versus $s$ at $N=8192$.}
\label{fig:phases}
\end{figure}

\subsection{ZGEMM performance and accuracy}
Table~\ref{tab:zgemm} compares the 3M emulation with native oneMKL ZGEMM
(0.76\,TFLOP/s, counting $8mnk$ flops). Since the emulation spends $3s$ INT8 GEMMs
on four times the flops of a real product, the emulated complex rate is
$\frac43$ of the real one: 24.4, 16.8 and 12.8\,TFLOP/s at $N=8192$ for $s=8$, 12
and 16, and up to 27.4\,TFLOP/s at $N=16384$ ($s=8$). At $s=16$ this is 13.7--17.8
times native ZGEMM for $N\ge4096$, and the board energy efficiency is 70 versus
9.4\,GFLOP/s/W. The accuracy behaves as for DGEMM (Table~\ref{tab:zaccuracy}): the
results are correctly rounded from $s=16$ on for $\mathcal{U}[-1,1)$ and from $s=24$
for $\phi=2$ (99.9\% at $s=20$). In the cancellation case the normwise error drops
from $3\times10^{-6}$ for native ZGEMM to $9\times10^{-16}$ at $s=20$. As
expected, the exact 3M form loses no accuracy relative to the real case.

\begin{table}[tbp]
\centering
\caption{Emulated ZGEMM (3M) versus native ZGEMM; rates in TFLOP/s ($8mnk$).}
\label{tab:zgemm}
\footnotesize
\begin{tabular}{rrrrrrr}
\toprule
$N$ & $s{=}8$ & $s{=}12$ & $s{=}16$ & $s{=}20$ & ZGEMM GPU & speed-up ($s{=}16$) \\
\midrule
1024 & 6.72 & 5.21 & 4.19 & 3.49 & 0.75 & 5.6 \\
2048 & 13.06 & 9.54 & 7.45 & 6.06 & 0.75 & 9.9 \\
4096 & 19.28 & 13.47 & 10.33 & 8.42 & 0.75 & 13.7 \\
6144 & 21.81 & 15.13 & 11.56 & 9.33 & 0.76 & 15.2 \\
8192 & 24.41 & 16.75 & 12.76 & 10.26 & 0.76 & 16.7 \\
12288 & 22.96 & 15.47 & 11.52 & -- & 0.76 & 15.1 \\
16384 & 27.36 & 18.00 & 13.60 & -- & 0.76 & 17.8 \\
\bottomrule
\end{tabular} {}
\end{table}

\begin{table}[tbp]
\centering
\caption{ZGEMM accuracy ($512\times512$, $k=4096$).}
\label{tab:zaccuracy}
\footnotesize
\begin{tabular}{llrrrrr}
\toprule
Data & Method & $t$ & max ulp & CR \% & max $\frac{|e|}{|A||B|}$ & $\frac{\|e\|_F}{\|C\|_F}$ \\
\midrule
U[-1,1) & ZGEMM GPU & 53 & $5.2\cdot 10^{6}$ & 1.87 & $2.4\cdot 10^{-16}$ & $3.0\cdot 10^{-15}$ \\
U[-1,1) & ZGEMM CPU & 53 & $6.5\cdot 10^{5}$ & 10.01 & $2.3\cdot 10^{-17}$ & $4.9\cdot 10^{-16}$ \\
U[-1,1) & Oz-II 3M $s{=}12$ & 39 & $1.8\cdot 10^{9}$ & 0.00 & $4.6\cdot 10^{-14}$ & $1.3\cdot 10^{-12}$ \\
U[-1,1) & Oz-II 3M $s{=}14$ & 47 & $3.7\cdot 10^{6}$ & 0.99 & $1.7\cdot 10^{-16}$ & $5.0\cdot 10^{-15}$ \\
U[-1,1) & Oz-II 3M $s{=}16$ & 55 & 0.50 & 100.00 & $3.1\cdot 10^{-18}$ & $4.7\cdot 10^{-17}$ \\
U[-1,1) & Oz-II 3M $s{=}20$ & 70 & 0.50 & 100.00 & $3.1\cdot 10^{-18}$ & $4.7\cdot 10^{-17}$ \\
U[-1,1) & Oz-II 3M $s{=}24$ & 85 & 0.50 & 100.00 & $3.1\cdot 10^{-18}$ & $4.7\cdot 10^{-17}$ \\
\midrule
$\phi$=2 & ZGEMM GPU & 53 & $2.5\cdot 10^{6}$ & 2.05 & $9.4\cdot 10^{-15}$ & $2.9\cdot 10^{-15}$ \\
$\phi$=2 & ZGEMM CPU & 53 & $8.0\cdot 10^{5}$ & 11.45 & $1.7\cdot 10^{-15}$ & $4.8\cdot 10^{-16}$ \\
$\phi$=2 & Oz-II 3M $s{=}12$ & 39 & $3.1\cdot 10^{11}$ & 0.00 & $1.3\cdot 10^{-9}$ & $8.4\cdot 10^{-11}$ \\
$\phi$=2 & Oz-II 3M $s{=}14$ & 47 & $1.4\cdot 10^{9}$ & 0.02 & $4.3\cdot 10^{-12}$ & $3.3\cdot 10^{-13}$ \\
$\phi$=2 & Oz-II 3M $s{=}16$ & 55 & $3.3\cdot 10^{6}$ & 3.76 & $1.2\cdot 10^{-14}$ & $1.3\cdot 10^{-15}$ \\
$\phi$=2 & Oz-II 3M $s{=}20$ & 70 & 32.98 & 99.91 & $9.1\cdot 10^{-17}$ & $4.6\cdot 10^{-17}$ \\
$\phi$=2 & Oz-II 3M $s{=}24$ & 85 & 0.50 & 100.00 & $9.1\cdot 10^{-17}$ & $4.6\cdot 10^{-17}$ \\
\midrule
cancel & ZGEMM GPU & 53 & $1.6\cdot 10^{17}$ & 0.00 & $1.9\cdot 10^{-16}$ & $3.2\cdot 10^{-6}$ \\
cancel & ZGEMM CPU & 53 & $1.9\cdot 10^{16}$ & 0.00 & $1.8\cdot 10^{-17}$ & $7.3\cdot 10^{-7}$ \\
cancel & Oz-II 3M $s{=}12$ & 39 & $1.1\cdot 10^{19}$ & 0.00 & $3.2\cdot 10^{-14}$ & $1.4\cdot 10^{-3}$ \\
cancel & Oz-II 3M $s{=}14$ & 47 & $1.1\cdot 10^{17}$ & 0.00 & $1.3\cdot 10^{-16}$ & $5.4\cdot 10^{-6}$ \\
cancel & Oz-II 3M $s{=}16$ & 55 & $7.7\cdot 10^{13}$ & 0.00 & $1.5\cdot 10^{-19}$ & $6.0\cdot 10^{-9}$ \\
cancel & Oz-II 3M $s{=}20$ & 70 & $2.2\cdot 10^{5}$ & 99.22 & $1.0\cdot 10^{-25}$ & $8.6\cdot 10^{-16}$ \\
cancel & Oz-II 3M $s{=}24$ & 85 & 14.45 & 100.00 & $1.7\cdot 10^{-27}$ & $4.7\cdot 10^{-17}$ \\
\bottomrule
\end{tabular} {}
\end{table}

\subsection{Energy}
During the emulation the board draws 120--193\,W, compared with 80\,W during native
DGEMM, but the run time is an order of magnitude shorter. At $N=8192$ the emulation
delivers 96 ($s=8$) and 50 ($s=16$) GFLOP/s per watt, against 9.5 for native DGEMM;
for $N\ge4096$ it thus needs 4.4--12 times less energy per product
(Table~\ref{tab:energy}). Per operation, an INT8 multiply--add pair costs 1.0\,pJ
on this device and an FP64 pair 100\,pJ (Table~\ref{tab:engines}).

\begin{table}[tbp]
\centering
\caption{Energy efficiency (board energy counter).}
\label{tab:energy}
\footnotesize
\begin{tabular}{rrrrrrr}
\toprule
$N$ & $s$ & W & J/call & GFLOP/s/W & DGEMM GPU & ratio \\
\midrule
1024 & 8 & 120 & 0.070 & 30.7 & 9.4 & 3.3 \\
1024 & 16 & 133 & 0.117 & 18.4 & 9.4 & 2.0 \\
2048 & 8 & 153 & 0.314 & 54.7 & 9.5 & 5.8 \\
2048 & 16 & 161 & 0.548 & 31.3 & 9.5 & 3.3 \\
4096 & 8 & 168 & 1.747 & 78.7 & 9.6 & 8.2 \\
4096 & 16 & 172 & 3.259 & 42.2 & 9.6 & 4.4 \\
8192 & 8 & 180 & 11.460 & 96.0 & 9.5 & 10.1 \\
8192 & 16 & 185 & 21.950 & 50.1 & 9.5 & 5.3 \\
16384 & 8 & 188 & 81.510 & 107.9 & 9.0 & 12.0 \\
16384 & 16 & 193 & 161.100 & 54.6 & 9.0 & 6.1 \\
24576 & 8 & 188 & 283.800 & 104.6 & 9.1 & 11.5 \\
24576 & 16 & 182 & 580.400 & 51.1 & 9.1 & 5.6 \\
\bottomrule
\end{tabular} {}
\end{table}

\subsection{Comparison with Ozaki scheme I on the same INT8 engine}
To quantify the benefit of scheme~II on this hardware, we implemented the INT8
Ozaki scheme~I \cite{ootomo2024dgemm,uchino2025ozaki1int8} with the same scaling,
the same INT8 GEMM and 7-bit balanced slices ($L$ slices, $L(L+1)/2$ products,
diagonal sums accumulated in double-double). At comparable accuracy, scheme~I needs
36--45 INT8 GEMMs where scheme~II needs 15--16, and it is 2.3--2.9 times slower
(Table~\ref{tab:ozaki1}, Figure~\ref{fig:oz1}). Because of its truncation it is also
less accurate than scheme~II at an equal number of GEMMs. We found no regime in
which scheme~I is preferable on the B60 and therefore include it only as a point of comparison.

\begin{table}[tbp]
\centering
\caption{Ozaki scheme I versus scheme II, $N=8192$, $\mathcal{U}[-1,1)$.}
\label{tab:ozaki1}
\footnotesize
\begin{tabular}{llrrrrrr}
\toprule
scheme & param. & bits & GEMMs & ms & TFLOP/s & max $\frac{|e|}{|A||B|}$ & CR \% \\
\midrule
Ozaki-I & $L{=}6$ & 40 & 21 & 168.8 & 6.51 & $2.2\cdot 10^{-13}$ & 0.0 \\
Ozaki-I & $L{=}7$ & 47 & 28 & 218.9 & 5.02 & $1.7\cdot 10^{-15}$ & 0.3 \\
Ozaki-I & $L{=}8$ & 54 & 36 & 276.0 & 3.98 & $1.7\cdot 10^{-17}$ & 29.6 \\
Ozaki-I & $L{=}9$ & 61 & 45 & 339.1 & 3.24 & $7.0\cdot 10^{-18}$ & 97.4 \\
Ozaki-II & $s{=}12$ & 39 & 12 & 91.2 & 12.06 & $6.3\cdot 10^{-14}$ & 0.0 \\
Ozaki-II & $s{=}14$ & 47 & 14 & 105.0 & 10.47 & $2.4\cdot 10^{-16}$ & 2.0 \\
Ozaki-II & $s{=}16$ & 55 & 16 & 118.7 & 9.26 & $7.0\cdot 10^{-18}$ & 100.0 \\
Ozaki-II & $s{=}18$ & 62 & 18 & 131.5 & 8.36 & $7.0\cdot 10^{-18}$ & 100.0 \\
\bottomrule
\end{tabular} {}
\end{table}

\begin{figure}[tbp]
\centering
\includegraphics[width=0.6\textwidth]{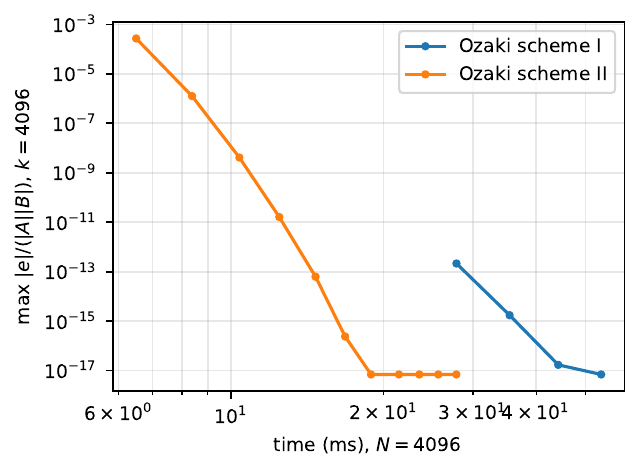}
\caption{Time ($N=4096$) versus componentwise error ($k=4096$) for both schemes.}
\label{fig:oz1}
\end{figure}

\subsection{Effect of the pipeline-level decisions}
Table~\ref{tab:ablation} disables the pipeline-level decisions of
Section~\ref{sec:arc} one at a time. Issuing the $s$ INT8 products as one batched
call saves 6--9\% at $N=2048$ and $N=24576$. The preference for power-of-two tiles
matters when tiling is active: at $N=24576$ and $s=16$, tiles with a
non-power-of-two row count cost 6\%. The padding of the leading dimension changes
these timings by at most 1\%, because the tilings chosen here rarely produce the
critical leading dimensions; its effect on the INT8 GEMM itself is shown in
Table~\ref{tab:lda}. At $N=8192$, where the problem fits into a single tile with
favourable dimensions, none of the decisions changes the run time.

\begin{table}[tbp]
\centering
\caption{Ablation of the pipeline-level decisions: run time of the full pipeline
and slowdown when one decision is disabled.}
\label{tab:ablation}
\footnotesize
\begin{tabular}{rrrrrr}
\toprule
 & & full pipeline & \multicolumn{3}{c}{slowdown without} \\
\cmidrule(l){4-6}
$N$ & $s$ & (ms) & batched GEMM & lda padding & $2^j$ tiles \\
\midrule
2048 & 8 & 2.0 & +6\% & 0 & 0 \\
2048 & 16 & 3.4 & +9\% & 0 & 0 \\
8192 & 8 & 63.6 & 0 & 0 & 0 \\
8192 & 16 & 118.7 & 0 & 0 & 0 \\
24576 & 8 & 1492.3 & +9\% & +1\% & 0 \\
24576 & 16 & 3200.3 & +7\% & 0 & +6\% \\
\bottomrule
\end{tabular} {}
\end{table}

\section{Discussion}
\label{sec:discussion}

\subsection{Using the library as a BLAS replacement}
The interface mirrors \texttt{dgemm} and \texttt{zgemm} of the oneMKL SYCL API with
a single extra argument, the number of moduli. We recommend the following policy.
By default we take $s=\texttt{moduli\_for\_bits}(k,53)$, i.e.\ $s=15$--$16$
for $k\le 2^{16}$; this is as accurate as DGEMM for well-scaled data and correctly
rounded for $\mathcal{U}[-1,1)$. Inputs with wide exponent ranges, or computations in which
cancellation is expected, call for $s=20$--$24$, while $s=8$--$12$ suffices where
an accuracy between FP32 and FP64 is acceptable. Because all accumulation is
performed in exact integer arithmetic, the results are deterministic and
independent of the tiling. Existing applications currently have to call the
library explicitly; a layer that transparently redirects their oneMKL calls is
left for future work.

\subsection{Performance model and outlook to Xe3/Xe3P}
\label{sec:model}
The measured run time is well described by
\begin{equation}
  T\approx\frac{2mnk\,s\,g}{R_8}+\frac{g\big[(16+s)(mk+kn)+(4s+8)mn\big]}{\beta_{\mathrm{mem}}}+n_{\mathrm{launch}}\tau,
\end{equation}
where $g=1$ for DGEMM and $g=3$ for ZGEMM, $R_8$ is the INT8 GEMM rate,
$\beta_{\mathrm{mem}}$ the memory bandwidth and $\tau\approx15\,\mu$s the launch
overhead. The three terms correspond to the three kinds of work in the pipeline.
The first is the time of the $s$ INT8 products of $2mnk$ operations each. The
second is the memory traffic of the remaining kernels, in bytes: the scaling and
split kernels read every element of $A$ and $B$ twice as FP64 (16 bytes) and write
$s$ residue bytes, and the CRT kernel reads $s$ INT32 products (4$s$ bytes) and
writes one FP64 result (8 bytes) per entry of $C$. The third counts the kernel
launches. With the B60 values of Table~\ref{tab:engines}, the model matches the
measurements to within 10\% on average and to within 1--8\% for
$4096\le N\le 8192$ (Table~\ref{tab:model}). It underestimates the time for small
problems, where the INT8 GEMM runs below its asymptotic rate, and for $N=12288$,
where oneMKL is slower for this shape. The fraction of time spent in INT8 GEMM is
approximately $\big(1+R_8[2(16+s)+(4s+8)]/(2sN\beta_{\mathrm{mem}})\big)^{-1}$.
Devices that raise $R_8$ faster than $\beta_{\mathrm{mem}}$ therefore shift the
balance towards the memory-bound phases, which our split and CRT kernels already
run at full bandwidth. For the D270 (Crescent Island), only the FP8 rate
(1.3\,PFLOP/s) and a memory bandwidth of $\approx$1.5\,TB/s have been reported
\cite{sth_d_series}. A quantitative projection requires the INT8 GEMM rate, which
has not yet been published; we therefore defer it.

Our design does not depend on INT8 specifically. Should future Intel GPUs favour
FP8 over INT8 in their matrix engines, as NVIDIA's Blackwell Ultra already does
\cite{nvidia_hgx}, the FP8 variant of Ozaki scheme~II \cite{uchino2026fp8} can be
mapped onto the same device-resident pipeline, with only the moduli and the
residue format changed.

\begin{table}[tbp]
\centering
\caption{Performance model versus measurement (single tile).}
\label{tab:model}
\footnotesize
\begin{tabular}{rrrrr}
\toprule
$N$ & $s$ & measured (ms) & model (ms) & deviation \\
\midrule
1024 & 8 & 0.58 & 0.46 & -21\% \\
1024 & 16 & 0.87 & 0.68 & -22\% \\
1024 & 24 & 1.21 & 0.90 & -26\% \\
2048 & 8 & 2.03 & 1.81 & -11\% \\
2048 & 16 & 3.38 & 3.05 & -10\% \\
2048 & 24 & 4.98 & 4.30 & -14\% \\
3072 & 8 & 5.21 & 4.73 & -9\% \\
3072 & 16 & 9.21 & 8.37 & -9\% \\
3072 & 24 & 13.58 & 12.01 & -12\% \\
4096 & 8 & 10.39 & 9.78 & -6\% \\
4096 & 16 & 18.89 & 17.73 & -6\% \\
4096 & 24 & 27.82 & 25.69 & -8\% \\
6144 & 8 & 30.23 & 28.50 & -6\% \\
6144 & 16 & 55.72 & 53.07 & -5\% \\
6144 & 24 & 82.74 & 77.64 & -6\% \\
8192 & 8 & 63.54 & 62.42 & -2\% \\
8192 & 16 & 118.73 & 117.95 & -1\% \\
8192 & 24 & 176.58 & 173.48 & -2\% \\
12288 & 8 & 222.23 & 193.62 & -13\% \\
12288 & 16 & 432.89 & 371.91 & -14\% \\
16384 & 8 & 432.86 & 438.95 & +1\% \\
\bottomrule
\end{tabular} {}
\end{table}

\subsection{Limitations}
\begin{itemize}
\item Exponents are shared per row and column, so elements more than $t-53$ binades
  below the row maximum lose bits (Section~\ref{sec:worst}); a larger $s$
  mitigates this.
\item The number of moduli is limited to $s\le24$ ($\log_2M<192$); $k>131071$ is
  handled by splitting $k$.
\item An Inf or NaN in a row of $\mathrm{op}(A)$ or a column of $\mathrm{op}(B)$
  produces NaN in the entire corresponding row or column of $C$, and subnormal
  results may be double-rounded.
\item The workspace grows as $s(mk+kn+4mn)$ bytes ($\times3$ for ZGEMM). Large
  problems are tiled, at the cost of re-splitting parts of $A$.
\item Small problems ($N\lesssim 256$) are dominated by launch overhead.
\end{itemize}

\section{Conclusions}
Ozaki scheme~II turns the INT8 matrix engines of an affordable Intel Arc Pro B60
into a DGEMM/ZGEMM engine that is an order of magnitude faster and about five times
more energy-efficient than the native FP64 path. With enough moduli it is also more
accurate: it rounds correctly on well-scaled data, where native DGEMM rounds only a
few percent of the entries correctly, and it is orders of magnitude more accurate
under cancellation. The same pipeline provides ZGEMM through an exact 3M scheme
with $3s$ instead of $4s$ INT8 products. On Xe2 the decisive implementation choices
were to keep FP64 out of the memory-bound kernels, to use an exact, division-free
CRT with an FP32 quotient, and to adapt the tiling to the behaviour of the vendor
INT8 GEMM. Future work includes measurements on Xe3P hardware, transparent
interposition into existing electronic-structure codes, batched and small-matrix
variants, and the reuse of the residues of $A$ across calls.

\section*{Code and data availability}
\label{sec:repro}
The library is available at \url{https://github.com/MacraMacra/OS2-Xe} under the
MIT license. It is header-only and consists of \texttt{ozaki2\_gpu.hpp} and the
table of moduli \texttt{coprime\_tables.hpp}; the repository also contains the
correctness tests, usage examples, and the throughput and accuracy benchmarks.

\section*{Acknowledgements}
This work was supported by Latvian Council of Science via Latvian State Research
Programme Nr.~VPPIZM-CERN-2022/1-0001.

\bibliographystyle{unsrt}

\end{document}